\documentclass[manuscript]{acmart}
\usepackage{subcaption}
\usepackage{multirow}
\usepackage{bbold}
\usepackage{balance}

\copyrightyear{2021}
\acmYear{2021}
\setcopyright{iw3c2w3}
\acmConference[]{}
\acmBooktitle{}
\acmPrice{}
\acmDOI{10.1145/3442381.3450052}
\acmISBN{978-1-4503-8312-7/21/04}

\AtBeginDocument{%
  \providecommand\BibTeX{{%
    \normalfont B\kern-0.5em{\scshape i\kern-0.25em b}\kern-0.8em\TeX}}}

\newcommand{\pos}{\text{pos}}
\newcommand{\wel}{\text{Wel}}
\newcommand{\rev}{\text{Rev}}

\newcommand{\opt}{\text{OPT}}
\newcommand{\welopt}{\text{Wel}^{\opt}}
\newcommand{\xhdr}[1]{\vspace{2mm} \noindent{\bf #1}}
\newcommand{\indic}{\mathbb{1}}

\newcommand{\edit}[1]{{{#1}}}

\ccsdesc[500]{Theory of computation~Algorithmic game theory and mechanism design}

\begin{document}

\title{Towards Efficient Auctions in an Auto-bidding World}

\author{Yuan Deng}
\email{dengyuan@google.com}
\affiliation{%
  \institution{Google}
  \streetaddress{111 8th Ave}
  \city{New York}
  \state{NY}
  \country{USA}
  \postcode{10011}
}
\author{Jieming Mao}
\email{maojm@google.com}
\affiliation{%
  \institution{Google}
  \streetaddress{111 8th Ave}
  \city{New York}
  \state{NY}
  \country{USA}
  \postcode{10011}
}
\author{Vahab Mirrokni}
\email{mirrokni@google.com}
\affiliation{%
  \institution{Google}
  \streetaddress{111 8th Ave}
  \city{New York}
  \state{NY}
  \country{USA}
  \postcode{10011}
}
\author{Song Zuo}
\email{szuo@google.com}
\affiliation{%
  \institution{Google}
  \streetaddress{111 8th Ave}
  \city{New York}
  \state{NY}
  \country{USA}
  \postcode{10011}
}

\begin{abstract}
Auto-bidding has become one of the main options for bidding in online advertisements, in which advertisers only need to specify high-level objectives and leave the complex task of bidding to auto-bidders. In this paper, we propose a family of auctions with boosts to improve welfare in auto-bidding environments with both return on ad spend constraints and budget constraints. Our empirical results validate our theoretical findings and show that both the welfare and revenue can be improved by selecting the weight of the boosts properly.
\end{abstract}

\keywords{Auto-bidding, Mechanism Design, Online Advertising}

\maketitle

\section{Introduction}

Auto-bidding, such as {\em Target CPA} (cost per acquisition) and {\em Target ROAS} (return on ad spend) has become one of the main options for bidding in online advertisements \cite{GoogleAutobidder,FacebookBid,TiktokBid}. Instead of asking the advertisers for fine-grained bids on potential keywords, auto-bidding products solicit high-level objectives and constraints only. Given the high-level information, the auto-bidding products bid on behalf of the advertisers at the serving time to convert the high-level information into a per-query bid, based on the predicted performance of the potential ad impression. For example, an auto-bidding product, such as target CPA, lets the advertisers to specify their daily budgets and average cost per conversion; with these information, it bids to maximize the number of conversions subject to these two constraints~\cite{GoogleAutobidder}. In this way, auto-bidding significantly simplifies the interaction between the platforms and the advertisers such that the advertisers can focus on their high-level goals and leave the complicated bidding part to the auto-bidding products.

The increasing popularity of auto-bidding not only opens up the opportunity of novel designs for auto-bidding products, but also calls a revisit of the current auction design for online advertising, which is heavily tailored for surplus-maximizing bidders~\cite{myerson1981optimal, edelman2007internet}. Surplus-maximizing bidders aim to optimize their quasi-linear utilities, i.e., total value minus total payments, but auto-bidders behave very differently. The most notable difference is that the objective of surplus-maximizing bidders takes the payment into account, while the payment only appear in the constraints of auto-bidders, e.g., target CPA auto-bidders. Therefore, classic results on auctions for surplus-maximizing bidders do not directly tell practitioners how to design auctions for auto-bidding environments. 

Recently, \citet{aggarwal2019autobidding} initiate the study of auction design with auto-bidding. They show that {\em uniform bidding} is the optimal bidding strategy for auto-bidders if the underlying auctions are incentive-compatible for surplus-maximizing bidders (e.g., second-price auctions). Here, uniform bidding is a bidding strategy in which the bidder always bids her private value multiplied by a constant multiplier across all auctions. Moreover, for single-slot auctions, they demonstrate that the welfare achieved by an incentive-compatible auction at any reasonable equilibrium is at least $1/2$ of the optimal welfare, and this $1/2$ approximation is tight for second-price auctions.

In this paper, we study how the auctioneer can improve the welfare efficiency of the auctions in auto-bidding environment using additive boosts. Roughly speaking, in such an auction, each bidder receives an additive boost separately that adds to her bid. Consequently a bidder receiving a larger boost would have a higher chance to win. Moreover, the boost is introduced in an incentive-compatible way
that the higher boost the winner receives, the less she needs to pay, if she wins the auction. 
The new auctions maintain the incentive compatibility if the underlying auctions are incentive compatible, and they can be integrated into most ranking-based auction systems easily.

\edit{
In practice, the boosts can be a quantitative measure that combines the advertiser value, the quality of the ads, and the seller cost. With the boost aligned with the mixed objective, the platform could obtain a better balance between these key metrics, leading to a better online advertising ecosystem in a long term.}

\subsection{Our results}
We mainly focus on the model of target ROAS bidder in this paper, since target CPA is a special case of target ROAS. A target ROAS bidder aims to maximize the total value from conversions subject to two constraints: the budget constraint and the return on ad spend constraint, i.e., the total value the bidder received should be no less than the total cost.\footnote{In general, the return on ad spend constraint specifies a minimum ratio between the total value and the total cost. Without loss of generality, we assume such a minimum ratio as $1$ throughout this paper.}

We first consider a simple case in which auto-bidders have return on ad spend constraints but no budget constraint. 
The boosts used by the auctioneer
is {\em $c$-value-competitive} (for some constant $c > 0$), if the difference between the boosts of any two bidders is at least $c$ times the difference between their values.

\begin{theorem}
\label{thm:no_budget_intro}
When auto-bidders have return on ad spend constraints but no budget constraint, Vickrey–Clarke–Groves (VCG) auctions with $c$-value-competitive boosts guarantees a $(c+1)/(c+2)$-approximation to the optimal welfare.
\end{theorem}

We emphasize our results only assume the bidders adopt feasible and undominated strategies, and therefore, our results do not rely on bidders bidding at any notions of equilibrium, in contrast to the classic notion of price of anarchy (PoA)~\citep{papadimitriou2001algorithms}. Our theorems can provide welfare approximation guarantees for practical scenarios, in which values are probably changing overtime so that auto-bidders usually do not converge to some equilibrium.

For auto-bidding with both return on ad spend constraints and budget constraints, we need a stronger notion of boosts, {\em $c$-benchmark-competitive boosts}.
The $c$-benchmark-competitive boosts are specified by a benchmark allocation and satisfy that the difference between the boosts of any two bidders is at least $c$ times the value of one of these two bidders who ranks higher in the benchmark.
\begin{theorem}
\label{thm:budget_intro}
When auto-bidders have both return on ad spend constraints and budget constraints, Vickrey–Clarke–Groves (VCG) auctions with $c$-benchmark-competitive boosts guarantees a $(c+1)/(c+2)$-approximation to the welfare in the benchmark allocation.
\end{theorem}

The welfare performance of $c$-benchmark-competitive boosts depends on the the welfare performance of the benchmark allocation. Particularly, if the benchmark allocation guarantees an $\gamma$-approximation to the optimal welfare, then the corresponding $c$-benchmark-competitive boosts guarantees a $\gamma \cdot (c+1)/(c+2)$-approximation to the optimal welfare.

In both Theorem~\ref{thm:no_budget_intro} and Theorem~\ref{thm:budget_intro}, the welfare approximation ratio approaches $1$ as $c$ goes to infinity; 
meanwhile, the revenue approaches $0$ since the boosts are deducted from the bidders' payment. Therefore, one must be cautious about the choice of $c$. As we will demonstrate in our empirical results, a properly selected $c$ can improve both revenue and welfare, but a choice of $c$ that is too large may have negative impact on revenue performance.

Under the assumption that auto-bidders always use uniform bidding, we show that all the above results can apply to generalized second-price (GSP) auctions. 
Without such an assumption, uniform bidding may not be the strategy adopted by auto-bidders, because uniform bidding may not be optimal strategy in GSP since GSP is not incentive compatible for surplus-maximizing bidders. We also provide a discussion about first-price auctions as the ad exchange market has been shifting to first-price auctions recently \cite{paes2020competitive}.

\subsection{Additional related work}

Optimal mechanism design is a central topic in economic study, dating back to the seminal work of Vickrey–Clarke–Groves (VCG) auctions~\citep{vickrey1961counterspeculation,clarke1971multipart,groves1973incentives} and Myerson's auction~\citep{myerson1981optimal}. Since then, mechanism design has been successfully deployed in many different fields, particularly in auctions, such as combinatorial auctions for reallocating radio frequencies~\citep{cramton2006combinatorial} and generalized second-price auctions and dynamic auctions for online advertising~\citep{edelman2007internet,mirrokni2020non}. In contrast to these works considering surplus-maximizing bidders, our work focuses on auto-bidders such as target CPA bidders and target ROAS bidders instead.

The simplest auto-bidding product before target CPA and target ROAS is budget optimization with surplus-maximizing bidders. The revenue-optimal auction of surplus-maximizing bidders with budget constraints is characterized by~\citet{pai2014optimal}. For applications in online advertising, \citet{balseiro2019learning} develop budget management strategies that are no-regret in a long run and   \citet{balseiro2017budget,balseiro2020budget} provide a thorough study to compare different commonly used budget management strategies in practice.

Our auctions with boosts are affine maximizer auctions ~\citep{lavi2003towards,sandholm2015automated} specialized in the ad auction setting. Intuitively, an affine maximizer auction specifies the seller's values for allocations, so that the auction finds an allocation that maximizes the sum of the bidders' values and the seller's value. It has been shown that any reasonable incentive-compatible and individual rational general combinatorial auction mechanism is an affine maximizer auction~\citep{lavi2003towards}.

Some of our results rely on the assumption that the auto-bidders adopt uniform bidding strategies when the underlying auction is not incentive-compatible for surplus-maximizing bidders. In practice, it is usually hard for bidders to adopt and optimize non-uniform bidding strategies, and moreover, uniform bidding has been shown to perform well against optimal non-uniform bidding strategies in ad auctions~\citep{feldman2007budget, feldman2008algorithmic, balseiro2019learning, deng2020data, bateni2014multiplicative}.

\section{Preliminaries}\label{sec:prelim}
We consider $n$ bidders bidding simultaneously in $m$ position auctions~\cite{lahaie2007sponsored,varian2007position}. We use $s_j$ to denote the number of slots of auction $j \in [m]$. For bidder $i \in [n]$, the value of the $k$-th slot in auction $j$ is $v_{i,j} \cdot \pos_{j,k}$. Here $v_{i,j}$ is the base value of auction $j$ for bidder $i$, and $\pos_{j,k}$ is the position normalizer of the $k$-th slot in auction $j$. Without loss of generality, we assume $\pos_{j,k}$ is non-increasing in $k$. We use $I = (n,m,\{s\}_j, \{v\}_{i,j}, \{\pos\}_{j,k})$ to denote a problem instance. As usual, we use $-i$ to denote the bidders other than bidder $i$.

We use bids $b$ to denote the bidders' bids such that $b_{i, j}$ is bidder $i$'s bid in auction $j$. Moreover, we use allocations $x$ and prices $p$ to denote the outcome of the auctions. In particular, $x_{i,j,k}$ is 1 if bidder $i$ gets the $k$-th slot of auction $j$ and 0 otherwise, and moreover, $p_{i,j}$ is the price paid by the bidder $i$ in auction $j$.

\xhdr{Position Auctions with Boosts.}
In the with-boost version, bidder $i$ in auction $j$ will receive an additive boost $z_{i,j} \geq 0$ such that her ranking score in the auction is the sum of her bid and her boost, i.e., $b_{i,j} + z_{i,j}$. We consider two types of auctions: the Vickrey–Clarke–Groves (VCG) auction and the generalized second-price (GSP) auction. Both of them generalize the second-price auction from single-slot auctions into position auctions. They have the same allocation rule which allocates the bidder with the $k$-th highest ranking score into the $k$-th slot.

In auction $j$, denote by $\hat{b}_{k,j}$ the $k$-th highest ranking score for $k \in [s_j]$. In a VCG auction, when the $k$-th slot's winner is bidder $i$, her payment in auction $j$ is $p_{i, j} = \sum_{\kappa = k+1}^{s_j} (\hat{b}_{\kappa,j} - z_{i,j})^+ \cdot (\pos_{\kappa-1,j} - \pos_{\kappa,j})$, where $(\cdot)^+$ denotes $\max\{0, \cdot\}$. In contrast, her payment is $p_{i, j} = (\hat{b}_{k+1,j} - z_{i,j})^+ \cdot \pos_{j,k}$ if she wins the $k$-th slot in a GSP auction. If bidder $i$ does not win any slot in auction $j$, her payment is simply $p_{i, j} = 0$. Notably, the boost $z_{i,j}$ is deducted from the runner-up's ranking score to preserve incentive compatibility.

\xhdr{Auto-bidders.}
All bidders considered in this paper are auto-bidders with budget constraints and return on ad spend constraints. In addition, we focus on auto-bidders adopting uniform bidding strategies since other strategies are dominated for truthful auctions\footnote{For GSP auctions, using non-uniform bidding strategies might be beneficial for the auto-bidders. We refer readers to Section~\ref{sec:non-truthful} for a detailed discussion.}~\citep{aggarwal2019autobidding}: each auto-bidder $i$ maintains a bid multiplier $\alpha_i$ and bids $b_{i,j} = \alpha_i \cdot v_{i,j}$ in auction $j$. The goal of the auto-bidder is to find $\alpha_i$ that maximizes the total value ($\sum_{j=1}^m \sum_{k=1}^{s_j} x_{i,j,k} \cdot v_{i,j} \cdot \pos_{j,k}$) subject to the constraint that the total payment ($\sum_{j=1}^m p_{i,j}$) is at most the minimum of the total value she received and the budget $B_i$.

Clearly, a bid multiplier would be invalid if it would result in violations of either the budget constraint or the return on ad spend constraint. Moreover, observe that using a bid multiplier strictly less than $1$ before the bidder hits her budget constraint is a dominated strategy.\footnote{This is because the payment is at most the bidder's bid in both GSP and VCG, and therefore, her total payment is strictly less than her total value when using a bid multiplier strictly less than $1$. Thus, the bidder can benefit from raising the bid multiplier.} We denote by $\Theta$ the set of multiplier vectors in which no bidder adopts a dominated bid multiplier or an invalid bid multiplier.

Now we can write the bid as a function of the multiplier vector $\alpha$. Moreover, the allocation $x$ and the prices $p$ can be written as functions of multiplier vector $\alpha$ and auction format $\mathcal A$ (e.g. VCG, GSP), i.e., $x(\alpha, \mathcal A)$ and $p(\alpha, \mathcal A)$.



\xhdr{Liquid welfare.} 
Our objective is to design auctions maximizing the liquid welfare: the sum of each buyers' received values capped by their budgets. Formally, the liquid welfare is given by 
\[
\wel(x) = \sum_{i \in [n]} \min\left(B_i, \sum_{j=1}^m \sum_{k=1}^{s_j} x_{i,j,k} \cdot v_{i,j} \cdot \pos_{j,k}\right).
 \]
 
Our benchmark is defined as the maximum liquid welfare among all possible allocations: $\welopt = \max_{x}\wel(x)$. Note that $\welopt$ is also the maximum revenue that can be extracted from auto-bidders. For each auction format $\mathcal A$, we measure the performance of our mechanisms by its competitive ratio, which is defined as 
\[
    \min_{\alpha \in \Theta} \frac{\wel(x(\alpha, \mathcal A)) }{\welopt}
\]
over all problem instances. Notice that our measure is different from the classic notion of price of anarchy (PoA), which restricts attention to bid multipliers forming Nash equilibrium~\citep{papadimitriou2001algorithms}. Here, we only assume the bidders are using strategies that are valid and undominated, while it is not necessary that their strategies must form a Nash equilibrium.
Finally, the revenue performance is simply defined as the sum of payments: $\rev = \sum_{j=1}^m \sum_{i \in [n]} p_{i,j}(\alpha, \mathcal A)$.

\edit{
\xhdr{Auction with seller cost.}
All our results can be generalized to an environment in which the seller has a selling cost $\psi_{i,j}$ for each bidder $i$ and auction $j$. To take the cost into account, the boosts $z$ are adapted in a way such that $\tilde{z}_{i,j} = z_{i, j} - \psi_{i, j}$ for all $i$ and $j$. We note that although this is a rather simple extension to our model, it could be of potential interests for practitioners because such seller costs could capture the quality of the bidder-auction matching, which might be essential for the long term consideration of the online advertising ecosystem.}

\section{Single-slot, No budgets}
\label{sec:warmup}
In this section, we study the simpler setting in which auto-bidders have return on ad spend constraints only (i.e. $B_i = +\infty$ for all $i \in [n]$) to establish intuitions on how boosts can help improve the welfare. Moreover, we assume all auctions are single-slot auctions so that both VCG auctions and GSP auctions become second-price auctions (i.e. $s_j = 1$ for all $j \in [m]$). Without loss of generality, the position normalizer in each auction is assumed to be $1$.

For simplicity, we first consider boosts that perfectly correlate with the bidders' valuations, denoted by {\em uniform boosts}. For uniform boosts, $z_{i,j} = c \cdot v_{i,j}$ for all bidder $i$ and auction $j$. Note that the uniform boost with weight $c$ is a special case of $c$-value-competitive boosts (cf. Theorem~\ref{thm:no_budget_intro} or Section~\ref{sec:uboosts}).

To establish intuitions, we first consider a simple example that demonstrates uniform boosts with a small boost weight would result in welfare inefficiency. In particular, uniform boosts of boost weight $0$ corresponds to plain second-price auctions without boosts. This example also provides an upper bound on the competitive ratio in second-price auctions with uniform boosts.

\begin{example}
\label{ex}
Consider a setting with two auto-bidders and two auctions (i.e. $n = 2$, $m= 2$). Moreover, assume both auctions are single-slot (i.e. $s_1= s_2=1$), and all position normalizers are $1$. The bidders' valuations are as follows:

\begin{center}
  \begin{tabular}{ccccc}
    \toprule
                & & Auction $1$     & & Auction $2$ \\
    \midrule
     Bidder $1$ & & $v_{1,1} = 1+d$ & & $v_{1,2} = \varepsilon$  \\
     Bidder $2$ & & $v_{2,1} = 0$   & & $v_{2,2} = 1$  \\
    \bottomrule
  \end{tabular}
\end{center}
Here $d > 0$ and $\varepsilon \in (0, 1)$ is a small number.
\end{example}

\begin{lemma}\label{lem:lower-bound-uboost}
For any $d > 0$ and an arbitrarily small $\varepsilon > 0$, in a second-price auction with uniform boosts of boost weight $c \leq d$, denoted by $\mathcal A$, there exists a bid multiplier vector $\alpha \in \Theta$ such that 
\[
\frac{\wel(x(\alpha, \mathcal A)) }{\welopt} \leq \frac{d+1+\varepsilon}{d+2}.
\]
Moreover, the bid multiplier vector $\alpha$ forms a Nash equilibrium.
\end{lemma}
\begin{proof}
First of all, observe that
\[
    \welopt = \max(v_{1,1},v_{2,1}) + \max(v_{1,2},v_{2,2}) = d+2.
\]
Consider a bid multiplier vector with $\alpha_1 > (1+c)/\varepsilon$ and $\alpha_2 = 1$. Under such a bid multiplier vector, the first bidder wins both auctions under $\mathcal A$ , and therefore $\wel(x(\alpha, \mathcal A)) = d+1+\varepsilon$. It is straightforward to verify that the return on ad spend constraints are satisfied: the first bidder receives total value $1 + d + \varepsilon$ and pays $1$. Moreover, for the second bidder to deviate, she has to use a bid multiplier $\alpha_2 \geq \alpha_1 \cdot \varepsilon > 1 + c$ to win the second auction, resulting in a payment $(1+c)$, which is larger than her received value $1$.
\end{proof}

When $c = d$ and $\varepsilon \to 0$, Lemma~\ref{lem:lower-bound-uboost} demonstrates a $(c+1)/(c+2)$ upper bound on the competitive ratio in second-price auctions with uniform boosts of boost weight $c$. Particularly, when $c = d = 0$ and $\varepsilon \to 0$, it recovers the $1/2$ upper bound for plain second-price auctions without boosts~\citep{aggarwal2019autobidding}.
The following theorem provides a matching lower bound on the competitive ratio in second-price auctions with uniform boosts of boost weight $c$.
\begin{theorem}\label{thm:uboost}
When auto-bidders have return on ad spend constraints only, the competitive ratio is at least $(c+1)/(c+2)$ in second-price auctions with uniform boosts of boost weight $c$.
\end{theorem}

\begin{proof}
Observe that the optimal allocation $x^{\opt}$ when auto-bidders have return on ad spend constraints only is simply $x^{\opt}_{i,j,1} = 1$ if and only if $v_{i,j} = \max_{i' \in [n]} v_{i', j}$, where ties are broken arbitrarily. Therefore, the optimal welfare is $\welopt = \sum_{j=1}^m \max_{i \in [n]} v_{i,j}$. 

Consider any multiplier vector $\alpha = (\alpha_1, \cdots ,\alpha_n) \in \Theta$ with $\alpha_i \geq 1$ for all $i$, and let $x$ be the resulting allocation of the second-price auctions with uniform boosts with boost weight $c$. Let $A$ to be the set of auctions $x$ agree with the optimal allocation $x^{\opt}$. In other words, $j \in A$ if and only if $x_{i,j,1} = x^{\opt}_{i,j,1}$ for all bidder $i$. 
By definition, we have
\begin{align*}
\wel(x) &= \sum_{j=1}^m \sum_{i \in [n]} v_{i,j} \cdot x_{i,j,1} \\
&=\sum_{j \in A} \max_{i \in [n]} v_{i,j} + \sum_{j \in [m] \setminus A} \sum_{i \in [n]} v_{i,j} \cdot x_{i,j,1}.
\end{align*}
For an auction $j \in [m] \setminus A$, let the winner under $x$ be bidder $i^*$ and let the winner under $x^\opt$ be bidder $i^\opt$. Then, bidder $i^*$'s payment in auction $j$ is
\begin{align*}
    (\hat b_{2, j} - z_{i^*, j})^+
    &{}\geq \alpha_{i^\opt} \cdot v_{i^\opt, j} + z_{i^\opt, j} - z_{i^*, j} \\
    &{}\geq v_{i^\opt, j} + c \cdot (v_{i^\opt, j} - v_{i^*, j}) \\
    &{}= (c + 1) \cdot \max_{i \in [n]} v_{i, j} - c \cdot \sum_{i \in [n]} v_{i, j} \cdot x_{i,j,1},
\end{align*}
where the first inequality follows that the second highest ranking score is not less than the ranking score of bidder $i^\opt$, and the second inequality is due to $\alpha_{i^\opt} \geq 1$ and the definition of the uniform boosts. 
Summing over all auctions in $B$, we have that the total revenue is at least
\[
\rev(x) \geq \sum_{j \in [m] \setminus A} \left( (c+1) \cdot \max_{i\in[n]} v_{i,j} - c \cdot \sum_{i\in[n]} v_{i,j} \cdot x_{i,j,1}\right).
\]
Observe that that $\wel(x) \geq \rev$ since no bidders pay more than their values, and therefore, we can conclude the proof:
\begin{align*}
\frac{c+2}{c+1} \cdot \wel(x) &{}\geq \wel(x) + \frac{\rev(x)}{c+1} \\
&{}\geq \sum_{j=1}^m \max_{i \in [n]} v_{i,j} 
   + \sum_{j \in [m] \setminus A}\sum_{i\in[n]} \left(1 - \frac{c}{c+1}\right) \cdot v_{i, j} \cdot x_{i,j,1} \\
&{}\geq \sum_{j=1}^m \max_{i \in [n]} v_{i,j} = \welopt.
\end{align*}
\end{proof}

\section{Value Competitive Boosts}
\label{sec:uboosts}
In this section, we show how to extend uniform boosts to multi-slot auctions and more flexible boost requirements, i.e., Theorem~\ref{thm:no_budget_intro}. 

As stated in the introduction, $c$-value-competitive boosts are boosts $z$ satisfying that $z_{i,j} - z_{i',j} \geq c \cdot (v_{i,j} - v_{i',j})$ for any pair of bidders $i$ and $i'$, and auction $j$ such that $v_{i,j} > v_{i',j}$.

\begin{theorem}[Restatement of Theorem \ref{thm:no_budget_intro}]\label{thm:no_budget_restate}
When auto-bidders have return on ad spend constraints but no budget constraints, Vickrey–Clarke–Groves (VCG) auctions with $c$-value-competitive boosts guarantees a $(c+1)/(c+2)$-approximation to the optimal welfare.
\end{theorem}

\begin{proof}
Consider any multiplier vector $\alpha = (\alpha_1, \cdots ,\alpha_n) \in \Theta$ with $\alpha_i \geq 1$ for all $i$, and let $x$ be the resulting allocation of the VCG auctions with $c$-value-competitive boosts $z$. Let $x^\opt$ be the allocation that gets the optimal welfare. For notation convenience, define $S_{j,k}$ and $S^\opt_{j,k}$ to be the set of bidders who get allocated to any of the top-$k$ slots in auction $j$ in allocation $x$ and $x^\opt$.

By definition, we have
\begin{align*}
\wel(x) &= \sum_{j=1}^m \sum_{k =1}^{s_j}\sum_{i \in [n]} v_{i,j}  \cdot \pos_{j,k} \cdot x_{i,j,k} \\
&= \sum_{j=1}^m \sum_{k =1}^{s_j}\sum_{i \in S_{j,k}} v_{i,j} \cdot (\pos_{j,k} - \pos_{j,k+1}).
\end{align*}
As for the revenue performance, observe that we have
\begin{align*}
\rev(x) &=  \sum_{j=1}^m \sum_{i \in [n]} p_{i,j} \\
&= \sum_{j=1}^m \sum_{k =1}^{s_j} \sum_{i \in S_{j,k}} (\hat{b}_{k,j} - z_{i,j})^+ \cdot (\pos_{j,k} - \pos_{j,k+1}) \\
&\geq \sum_{j=1}^m \sum_{k =1}^{s_j} (\pos_{j,k} - \pos_{j,k+1}) \cdot \sum_{i \in S_{j,k} \setminus S^\opt_{j,k}} (\hat{b}_{k,j} - z_{i,j})^+.
\end{align*}
For auction $j$, the $k$-th highest ranking score $\hat{b}_{k,j}$ should be higher than the ranking scores from bidders who do not win any of the top-$k$ slots. Formally, 
\[
    \forall i \not\in S_{j,k}, \hat{b}_{k,j} \geq \alpha_i \cdot v_{i,j} + z_{i,j} \geq v_{i,j} + z_{i,j}.
\]
As a result, for any auction $j$ and slot $k \in [s_j]$, we have
\begin{align*}
   & \sum_{i \in S_{j,k} \setminus S^\opt_{j,k}} (\hat{b}_{k,j} - z_{i,j})^+\\
   \geq{}& | S_{j,k} \setminus S^\opt_{j,k}| \cdot \hat{b}_{k,j} - \sum_{i \in S_{j,k} \setminus S^\opt_{j,k}} z_{i,j}\\
   ={}& | S^\opt_{j,k} \setminus S_{j,k}| \cdot \hat{b}_{k,j} - \sum_{i \in S_{j,k} \setminus S^\opt_{j,k}} z_{i,j}\\
   \geq{}& \sum_{i \in S^\opt_{j,k} \setminus S_{j,k}} (v_{i,j}+z_{i,j})- \sum_{i \in S_{j,k} \setminus S^\opt_{j,k}} z_{i,j}\\
   \geq{}& \sum_{i \in S^\opt_{j,k} \setminus S_{j,k}} (1+c)\cdot v_{i,j}- \sum_{i \in S_{j,k} \setminus S^\opt_{j,k}} c\cdot v_{i,j}.
\end{align*}
The second equation above comes from the fact that $|S_{j,k}| = |S^\opt_{j,k}| = k$ implies $| S_{j,k} \setminus S^\opt_{j,k}|=| S^\opt_{j,k} \setminus S_{j,k}|$. The second inequality follows $\hat b_{k,j} \geq v_{i,j} + z_{i,j}$ for $i \not \in S_{j, k}$. The last inequality follows the definition of $c$-value-competitive boosts: Note that for any $i\in S^\opt_{j,k} \setminus S_{j,k}$ and any $i' \in S_{j,k} \setminus S^\opt_{j,k}$, $v_{i, j} \geq v_{i',j}$, so $z_{i,j} - z_{i',j} \geq c \cdot (v_{i, j} - v_{i',j})$.

Combining everything,
\begin{align*}
    &{}\frac{c+2}{c+1} \cdot \wel(x) \geq \wel(x) + \frac{\rev(x)}{c+1} \\
    \geq{}& \sum_{j=1}^m \sum_{k =1}^{s_j} (\pos_{j,k} - \pos_{j,k+1}) \cdot\\ &{}\left(\sum_{i \in S^\opt_{j,k} \cap S_{j,k}} v_{i,j} + \sum_{i \in S^\opt_{j,k} \setminus S_{j,k}} v_{i,j} + \left(1-\frac{c}{c+1}\right) \sum_{i \in S_{j,k} \setminus S^\opt_{j,k}} v_{i,j} \right)\\
    \geq{}& \sum_{j=1}^m \sum_{k =1}^{s_j} (\pos_{j,k} - \pos_{j,k+1}) \cdot \sum_{i \in S^\opt_{j,k}} v_{i,j} = \welopt,
\end{align*}
which concludes the proof.
\end{proof}

\section{Benchmark Boosts}
\label{sec:benchmark_boost}
In this section, we switch to focus on general cases by considering auto-bidders with both return on ad spend constraints and budget constraints. The $c$-benchmark-competitive boosts are computed with respect to a benchmark $o$. A benchmark $o$ is a collection of rankings for auctions, such that $o_j: [n] \to [n]$ is a permutation of bidders in auction $j$. The benchmark allocation induced from a benchmark $o$ allocates the $k$-th slot in auction $j$ to bidder $o_j(k)$. With a slight abuse of notation, the welfare performance of a benchmark $o$ is defined as
\[
    \wel(o) = \sum_{i \in [n]} \min\left(B_i, \sum_{j=1}^m \sum_{k=1}^{s_j} \indic\{i = o_j(k)\} \cdot v_{i,j} \cdot \pos_{j,k}\right),
\]
where $\indic \{\cdot\}$ is the indicator function. 
Given a benchmark $o$, $c$-benchmark-competitive boosts are boosts $z$ satisfying
\[
    z_{o_j(k),j} - z_{o_j(k'),j} \geq c \cdot v_{o_j(k), j},
\]
for all $k \in [s_j]$, $k' > k$, and auction $j$.

\begin{theorem}[Restatement of Theorem \ref{thm:budget_intro}]

\label{thm:benchmark}
When auto-bidders have both return on ad spend constraints and budget constraints, Vickrey–Clarke–Groves (VCG) auctions with $c$-benchmark-competitive boosts guarantees a $(c+1)/(c+2)$-approximation to $\wel(o)$.
\end{theorem}
\begin{proof}
Consider any multiplier vector $\alpha = (\alpha_1, \cdots ,\alpha_n) \in \Theta$, and let $x$ be the resulting allocation of the VCG auctions with $c$-value-competitive boosts. Let $M$ be the set of auto-bidders whose payments meet their budgets. For auto-bidder $i \in [n] \setminus M$ not meeting the budget, we have $\alpha_i \geq 1$. Moreover, let $x^o$ be the allocation induced from the benchmark $o$. For notation convenience, define $S_{j,k}$ and $S^o_{j,k}$ to be the set of bidders who get allocated to any of the top-$k$ slots in auction $j$ in allocation $x$ and $x^o$, respectively. By definition, we have
\begin{align*}
\wel(x) &=\sum_{i \in [n]} \min\left(B_i, \sum_{j=1}^m \sum_{k =1}^{s_j}v_{i,j}  \cdot \pos_{j,k} \cdot x_{i,j,k}\right) \\
&= \sum_{i\in M} B_i + \sum_{j=1}^m \sum_{k =1}^{s_j}\sum_{i \in S_{j,k}\setminus M} v_{i,j} \cdot (\pos_{j,k} - \pos_{j,k+1}).
\end{align*}
As for the revenue performance, observe that we have
\begin{align*}
\rev(x) &=  \sum_{j=1}^m \sum_{i \in [n]} p_{i,j} \\
&= \sum_{j=1}^m \sum_{k =1}^{s_j} \sum_{i \in S_{j,k}} (\hat{b}_{k,j} - z_{i,j})^+ \cdot (\pos_{j,k} - \pos_{j,k+1}) \\
&\geq \sum_{j=1}^m \sum_{k =1}^{s_j} (\pos_{j,k} - \pos_{j,k+1}) \cdot \sum_{i \in S_{j,k} \setminus S^o_{j,k}} (\hat{b}_{k,j} - z_{i,j}).
\end{align*}

For auction $j$, the $k$-th highest ranking score $\hat{b}_{k,j}$ should be higher than the ranking scores from bidders who do not win any of the top-$k$ slots. In other words, $\forall i \not\in S_{j,k}$, $\hat{b}_{k,j} \geq \alpha_i \cdot v_{i,j} + z_{i,j}$. Similar to the proof for Theorem~\ref{thm:no_budget_restate}, we have that $|S^o_{j,k}|=| S_{j,k}| = k$, which implies that $|S^o_{j,k} \setminus S_{j,k}| = |S_{j,k} \setminus S^o_{j,k}|$. Therefore, for any auction $j$ and slot $k \in [s_j]$,
\begin{align*}
   \sum_{i \in S_{j,k} \setminus S^o_{j,k}} (\hat{b}_{k,j} - z_{i,j}) \geq   \sum_{i' \in S^o_{j,k} \setminus S_{j,k}} (\alpha_{i'} v_{i',j} + z_{i',j})-\sum_{i \in S_{j,k} \setminus S^o_{j,k}} z_{i,j}.
\end{align*}
For each $i' \in S^o_{j,k}$ and $i \not \in S^o_{j,k}$, $i'$ ranks ahead of $i$ in the benchmark $o$. By the definition of $c$-benchmark-competitive boosts, we have
$z_{i',j} - z_{i,j} \geq c\cdot v_{i',j}$,
and therefore,
\[
   \sum_{i \in S_{j,k} \setminus S^o_{j,k}} (\hat{b}_{k,j} - z_{i,j}) \geq  \sum_{i \in S^o_{j,k} \setminus S_{j,k}} (\alpha_i + c) \cdot v_{i,j}.
\]
Plugging this back into $\rev(x)$, we have
\begin{align*}
    \rev(x)&\geq\sum_{j=1}^m \sum_{k =1}^{s_j} (\pos_{j,k} - \pos_{j,k+1}) \cdot \sum_{i \in S^o_{j,k} \setminus S_{j,k}} (\alpha_i + c) \cdot v_{i,j} \\&\geq\sum_{j=1}^m \sum_{k =1}^{s_j} (\pos_{j,k} - \pos_{j,k+1}) \cdot \sum_{i \in (S^o_{j,k} \setminus S_{j,k}) \setminus  M} (1 + c) \cdot v_{i,j},
\end{align*}
where the last inequality follows that $\alpha_i \geq 1$ for $i \in [n] \setminus M$.
Combining everything together, we finally have
\begin{align*}
    &~\frac{c+2}{c+1} \cdot \wel(x) \geq \wel(x) + \frac{\rev(x)}{c+1} \\
    \geq&~ \sum_{i\in M} B_i+ \sum_{j=1}^m \sum_{k =1}^{s_j} (\pos_{j,k} - \pos_{j,k+1}) \cdot\\ &\left[\left(\sum_{i \in S_{j,k}\setminus M } v_{i,j}\right) + \left(\sum_{ i \in (S^o_{j,k} \setminus S_{j,k}) \setminus  M} v_{i,j}\right)  \right]\\
    \geq&~  \sum_{i \in M} B_i + \sum_{j=1}^m \sum_{k =1}^{s_j} (\pos_{j,k} - \pos_{j,k+1}) \cdot \sum_{i \in S^o_{j,k} \setminus M} v_{i,j} \\
    \geq& \sum_{i=1}^n \min\left(B_i, \sum_{j=1}^m \sum_{k =1}^{s_j} (\pos_{j,k} - \pos_{j,k+1}) \sum_{i \in S^o_{j,k}} v_{i,j}\right) = \wel(o),
\end{align*}
which concludes the proof.
\end{proof}

The minimal boosts that are $c$-benchmark-competitive with respect to a benchmark $o$ are 
\[
    z_{o_j(k),j} = c \cdot \sum_{\kappa = k}^{s_j} v_{i,o_j(\kappa)},
\]
for all $k \in [s_j]$ and auction $j$. For $k > s_j$, $z_{o_j(k) ,j}$ is simply $0$. We denote such boosts by {\em benchmark boosts}. Intuitively, for a bidder $i$ who ranks $k \leq s_j$ in auction $j$ in benchmark $o$, her boost should be $c$ times the sum of the values of the bidders ranked below her in auction $j$ in benchmark $o$ (including herself).


\section{Non-truthful Auctions}\label{sec:non-truthful}

As we mentioned in Section~\ref{sec:prelim} that in truthful auctions, any non-uniform bidding strategy is dominated by uniform bidding strategies. As a result, it is reasonable to focus on uniform bidding strategies only, which are easy to implement and guarantee optimality for auto-bidders. However, in non-truthful auctions, the optimal non-uniform bidding strategy may dominate uniform strategies. 

Therefore, in non-truthful auctions such as GSP or first-price auctions, it could be beneficial for auto-bidders to adopt non-uniform bidding strategies despite of the additional cost from implementation complexity. In this section, we discuss under what condition and to what extend, our theorems for VCG auctions could generalize to non-truthful auctions.

\subsection{Generalized Second-Price (GSP) auction}

GSP auctions are commonly used in online advertising when there are multiple slots, yet it is widely known that it is not incentive compatible \cite{edelman2007internet}. The following example shows that a non-uniform bidding strategy outperforms any uniform bidding strategies for auto-bidders with return on ad spend constraints.

\begin{example}
\label{ex_gsp}
Consider a setting with three auto-bidders and two auctions (i.e. $n = 3$, $m= 2$). The first auction has $2$ slots with position normalizer $1$ and $0.9$. The second auction has $1$ slot with position normalizer $1$. The bidders' valuations are as follows:

\begin{center}
  \begin{tabular}{ccccc}
    \toprule
                & & Auction $1$     & & Auction $2$ \\
    \midrule
     Bidder $1$ & & $v_{1,1} = 1$ & &  $v_{1,2} = 0.8$  \\
     Bidder $2$ & & $v_{2,1} = 0.9$   & & $v_{2,2} = 0$  \\
     Bidder $3$ & & $v_{2,1} = 0$   & & $v_{2,2} = 1$  \\
    \bottomrule
  \end{tabular}
\end{center}
\end{example}
In this example, bidder $2$ and $3$ are only interested in one auction and assume that they use bid multiplier $1$. First we consider the cases when bidder $1$ uses a uniform bid multiplier $\alpha_1$.
\begin{itemize}
    \item Bidder $1$ does not win auction $2$: in this case, bidder $1$ gets at most value $1$ from auction $1$.
    \item Bidder $1$ wins auction $2$: in this case, $\alpha_1$ needs to be at least $1.25$. Bidder $1$ will win the first slot of auction $1$ with payment $0.9$, and win the slot of auction $2$ with payment $1$. In this case, bidder $1$ pays $1.9$ to get total value $1.8$ and violates the return on ad spend constraint.
\end{itemize}
To sum up, bidder $1$ with a uniform bid multiplier can only win auction $1$ to get value at most $1$. On the other hand, consider the case when bidder $1$ uses an arbitrarily small bid multiplier $\varepsilon > 0$ in auction $1$ and a bid multiplier at least $1.25$ in auction $2$. In this case, bidder $1$ wins the second slot of auction $1$ and slot of auction $2$, and pays in total $1$. In this situation, bidder $1$ gets total value $1.7~(=1 \cdot 0.9 + 0.8 \cdot 1)$. Therefore, bidder $1$ can get more value by non-uniform bidding.


Although uniform bidding is no longer the optimal strategy for auto-bidders in GSP auctions, generalizations of our results are still possible if a minimum bidding level $\gamma > 0$ is guaranteed for all bidders across all auctions, i.e., 
\[\forall i \in [n], j \in [m], ~ b_{i,j} \geq \gamma \cdot v_{i, j}.\]

\begin{lemma}\label{lem:general}
  Theorem~\ref{thm:uboost}, Theorem~\ref{thm:no_budget_restate}, and Theorem~\ref{thm:benchmark} can be generalized to auction format $\mathcal{A}$ with boosts of boost weight $c$ accordingly with approximation ratio $(c + \gamma) / (c + \gamma + 1)$, if the following properties are met:
  \begin{itemize}
    \item $\gamma > 0$ is the minimum bidding level;
    \item The allocation of $\mathcal{A}$ is the same as VCG auctions with boosts;
    \item The payment of $\mathcal{A}$ is no less than the payment of VCG auctions with boosts on all auction inputs, and no more than the bid, i.e., for a bidder $i$ winning the $k$-th slot in auction $j$, her payment $p_{i,j}$ satisfies the criteria that
    \[
        \sum_{\kappa = k+1}^{s_j} (\hat{b}_{\kappa,j} - z_{i,j})^+ \cdot (\pos_{\kappa-1,j} - \pos_{\kappa,j}) \leq p_{i,j} \leq b_{i,j}.
    \]
  \end{itemize}
\end{lemma}

We omit the proof of Lemma~\ref{lem:general} but highlight that comparing with all previous theorems: (i) the welfare on the same instance remains the same for $\mathcal{A}$, because the allocation does not change; (ii) the revenue on the same instance is weakly higher for $\mathcal{A}$, because the payment weakly increases; (iii) in previous proofs, we apply $b_{i,j} = \alpha_i \cdot v_{i,j} \geq v_{i,j}$ for bidder $i$ with bid multiplier $\alpha_i \geq 1$, which can be replaced by $b_{i,j} \geq \gamma \cdot v_{i,j}$ in the proof of Lemma~\ref{lem:general}.

\xhdr{Enforcing uniform bidding.}
Observe that $\gamma \geq 1$ in GSP auctions when uniform bidding is enforced, so we have the following corollary that generalizes all our previous theorems for GSP auctions.

\begin{corollary}
  When uniform bidding strategies are enforced, Theorem~\ref{thm:uboost}, Theorem~\ref{thm:no_budget_restate}, and  Theorem~\ref{thm:benchmark} can be generalized to GSP auctions with boosts.
\end{corollary}

In fact, uniform bidding is a commonly used strategy with good approximation guarantees \cite{feldman2007budget, feldman2008algorithmic,bateni2014multiplicative,deng2020data}. Non-uniform bidding strategies, in this case, could be less beneficial for auto-bidders because the potential improvement is limited or dominated by the additional cost from the implementation complexity.

\subsection{First-price auctions}

First-price auctions, which has been widely adopted by the ad exchange recently, is another important non-truthful auction in the industry~\citep{paes2020competitive}. The next example demonstrates a case in which a non-uniform bidding strategy outperforms uniform bidding strategies.

\begin{example}
  Consider a setting with two auto-bidders and two auctions (i.e., $n = 2$, $m = 2$). Both auctions have only one slot, $s_1 = s_2 = 1$. The bidders' valuations are as follows:
  \begin{center}
  \begin{tabular}{ccccc}
    \toprule
                & & Auction $1$    & & Auction $2$ \\
    \midrule
     Bidder $1$ & & $v_{1,1} = 4$ & & $v_{1,2} = 1$  \\
     Bidder $2$ & & $v_{2,1} = 1$ & & $v_{2,2} = 2$  \\
    \bottomrule
  \end{tabular}
\end{center}
\end{example}
In this example, when both bidders adopt uniform bidding, they must have bid multipliers being $1$ to cope with the return on ad spend constraint. In this case, they receive values $4$ and $2$, respectively. However, if bidder $1$ bids $2$ and $3$ in these two auctions, respectively, she can receive a total value $5$. Note that this is an equilibrium because if bidder $2$ raises her bids to win either or both auction, her return on ad spend constraint will be violated.

Interestingly, if uniform bidding strategies are enforced under first-price auction, then having $\alpha_i = 1$ is the optimal strategy for auto-bidding environments with return on ad spend constraints but without budget constraint, which then simply leads to an equilibrium achieving the optimal welfare and revenue. 

\begin{theorem}\label{thm:fpa}
  When uniform bidding strategies are enforced, the optimal welfare and revenue are achieved under any Nash equilibrium of auto-bidders with return on ad spend constraints but without budget constraint in first-price auctions.
\end{theorem}

\begin{proof}
  If there is any auction $j$ such that the allocation of the auction is not optimal under the equilibrium, then there must be an auto-bidder $i$ who wins slot $k \in [s_j]$ in the optimal allocation, but does not win any slot higher than $k$ in the equilibrium. Note that overbidding in first-price auction leads to violation of the return on ad spend constraint, so the bid multiplier for any auto-bidder is no more than $1$ in any equilibrium. Therefore, auto-bidder $i$ can raise her bid multiplier to $1$ and win slot $k$ in auction $j$. Moreover, observe that raising her bid multiplier will not violate her return on ad spend constraint but strictly improves her total value, contradicting with the assumption of Nash equilibrium. 
\end{proof}

Despite of the desirable guarantee from Theorem~\ref{thm:fpa}, enforcing uniform bidding in the first-price auction could be a rather strong assumption comparing with the same assumption for GSP auctions~\citep{deng2020data}. One reason is that non-uniform bidding in this case may bring up significant improvements, dominating the additional cost from its implementation complexity. In fact, more and more such researches have emerged following the ad exchange market shifting to first-price auctions \cite{morishita2020online,paes2020competitive}.

\section{Experiments}

In this section, we validate our theoretical findings with semi-synthetic data derived from real auction data of a major search engine. The main import of using real ad auction data is that the data captures variation across bidders. Note that when the bidders are symmetric (i.e., their value distributions are i.i.d. across bidders), optimal efficiency is achieved in any symmetric equilibrium, and therefore, no efficiency improvement can be observed by applying our mechanisms. We simulate both VCG and GSP auctions with bids from auto-bidders only. Instead of using the real constraints, we generate artificial budgets and return on ad spend targets, which excludes any practical noises from the real ad system. 
We emphasize that our main objective for conducting experiments is
to validate our theoretical findings rather than investigating the
efficiency potentials on an auction system actually implemented in
practice. 
In practice, ad auction systems often need to take care of many practical aspects that would never be considered in theory.

\subsection{Experiment Setup}

\xhdr{Optimal Benchmark.}
To avoid getting distracted from the computation of the optimal benchmark, we first define the benchmark allocation and then derive the corresponding budgets of the bidders from the benchmark allocation so that the benchmark allocation would be optimal. In particular, we select a random subset of bidders and mark them as budget-constrained. For each budget-constrained bidder, we assign a random factor $\mu_i \in (0, 1)$, and we assign $\mu_i = 1$ for bidders without budget constraints. Given a vector of $\mu$, the optimal benchmark allocation in auction $j$ is obtained by ranking bidders according to $\mu_i \cdot v_{i,j}$. Finally, for the selected bidders that are budget-constrained, we set their budgets equal to their total received values under the benchmark allocation. Intuitively, $\mu_i$ for bidder $i$ can be seen as one minus the Lagrange multiplier for her budget constraint in the liquid welfare optimization program. In this way, we augment the auction dataset with artificial budgets in which the optimal benchmark allocation is known.

\xhdr{Simulation Procedures.} To properly evaluate the efficiency of a new mechanism, after the generation of the dataset, we first pre-train the bid multipliers $\alpha$ using the corresponding auction format (VCG or GSP) with no boosts in $25$ iterations to obtain an equilibrium as a starting point.

We simulate the response of auto-bidders by gradient descent on their bid multipliers in log space until convergence~\citep{nesterov2013introductory,aggarwal2019autobidding}. Formally, let $\alpha_{i, t}$ be the bid multiplier for bidder $i$ in iteration $t$. Moreover, let $\mathsf{target~spend}_{i, t}$ be the minimum between bidder $i$'s budget and her total received value in iteration $t$, and let $\mathsf{spend}_{i, t}$ be bidder $i$'s total payment in iteration $t$. Then, the bidder $i$'s bid multiplier in iteration $t+1$ is updated by
\begin{align*}
  \log \mathsf{\alpha}_{i, t+1} = (1-\eta_t) \cdot \log \mathsf{\alpha}_{i, t} + \eta_t \cdot \log \frac{\mathsf{target~spend}_{i, t}}{\mathsf{spend}_{i, t}},
\end{align*}
where $\eta_t \in (0, 1)$ is a properly chosen learning rate for iteration $t$. Intuitively, bidder $i$'s bid multiplier increases for the next iteration if $\mathsf{spend}_{i, t} < \mathsf{target~spend}_{i, t}$; otherwise if $\mathsf{spend}_{i, t} > \mathsf{target~spend}_{i, t}$, her bid multiplier decreases for the next iteration.

After obtaining a starting point, we simulate another $25$ iterations for auctions with boosts. In this way, we can observe both the initial impact of adding the boosts and how the impact changes overtime after auto-bidders' response. 

\xhdr{Boosts.} In addition to the baseline in which we continue to use auctions without boosts, we experiment with uniform boosts (Section~\ref{sec:warmup}) and benchmark boosts (Section~\ref{sec:benchmark_boost}) with different boost weights $c$, denoted by $\text{uboost-}c$ and $\text{benchmark-}c$, respectively. Recall that uniform boosts are designed for environments without budget constraints while benchmark boosts can accommodate budget constraints, and therefore, we expect benchmark boosts outperform uniform boosts in our experiments. 

All metrics we report are relative to the baseline (i.e., no boosts); therefore, the metrics of the baseline are normalized to 1.

\subsection{Experimental Results}

\begin{figure}[h]
    \centering
    \begin{subfigure}[b]{0.4\textwidth}
        \centering
        \includegraphics[width=\textwidth]{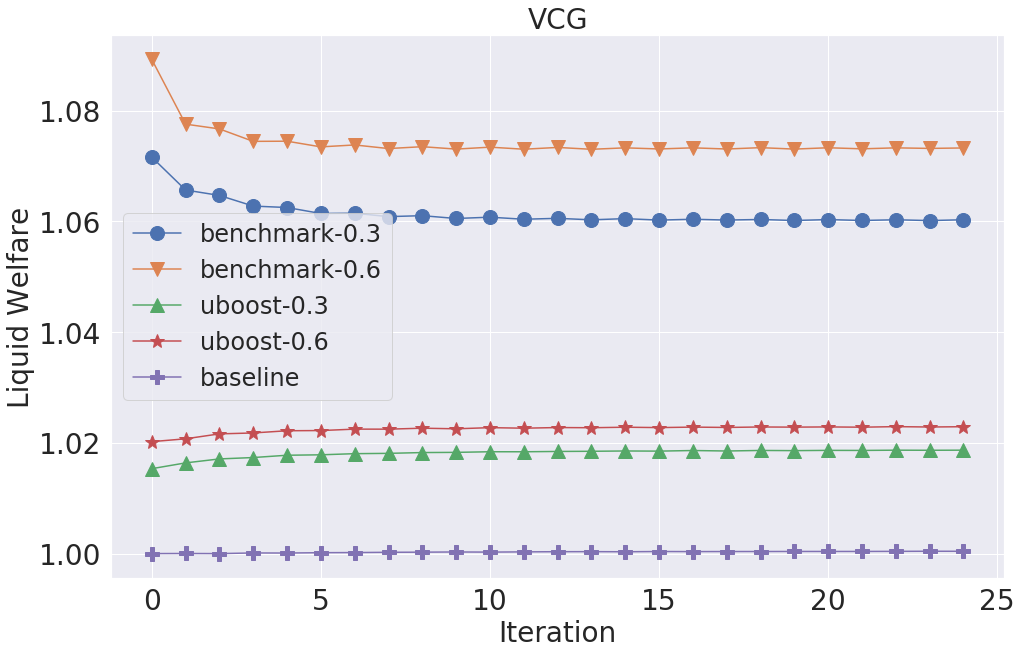}
        \caption{Liquid Welfare}
        \label{fig:VCG_liquid}
    \end{subfigure}
    \qquad
    \begin{subfigure}[b]{0.4\textwidth}
        \centering
        \includegraphics[width=\textwidth]{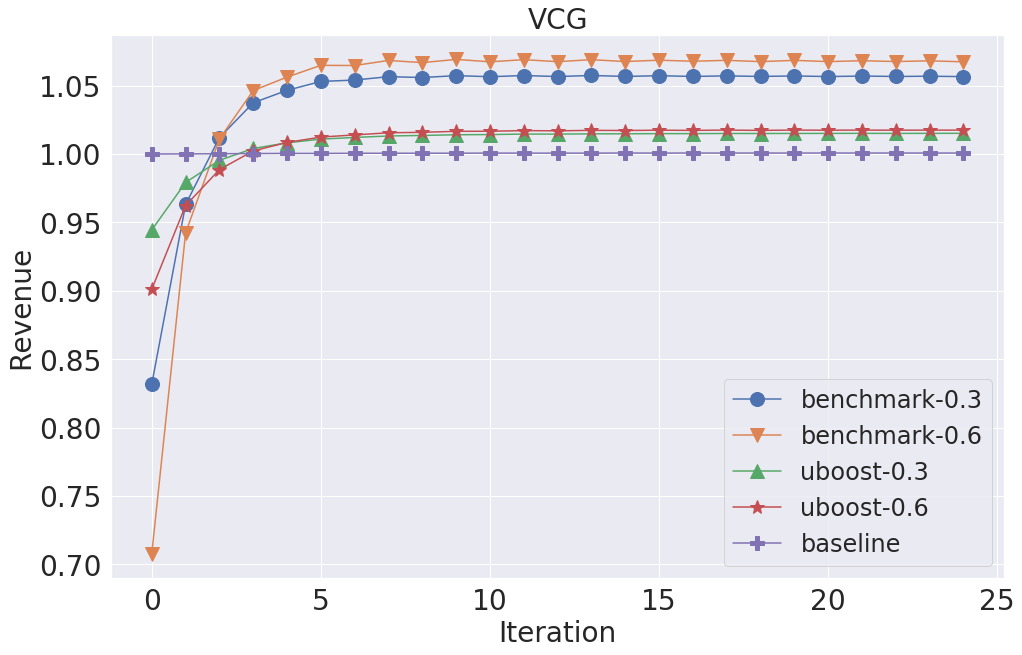}
        \caption{Revenue}
        \label{fig:VCG_rev}
    \end{subfigure}
    \caption{Welfare and revenue performance of uniform boosts and benchmark boosts for each iteration in VCG.}\label{fig:VCG_trend}
\end{figure}

Figure~\ref{fig:VCG_liquid} reports the trend of welfare performance under uboost-$0.3$, uboost-$0.6$, benchmark-$0.3$, and benchmark-$0.6$ in VCG auctions. Note that all of these boosts have positive initial impact on welfare, which verifies that our welfare guarantees do not rely on auto-bidders bidding in equilibrium. Moreover, as predicted, benchmark boosts have better welfare improvement than the uniform boosts, and as the boost weight $c$ increases, the welfare improvement also increases. After the bidders' response, we observe that the welfare improvements of benchmark boosts slightly decrease while the welfare improvements of uniform boosts slightly increase, but their relative ordering remains the same. We only report the trends for VCG auctions and the trends for GSP auctions are similar.

As for revenue performance, Figure~\ref{fig:VCG_rev} shows a negative initial impact on revenue for all boosts. This is because the boosts are deducted from the bidders' payments, and therefore, we observe that the larger the boost weight is, the stronger the negative impact of revenue is in the first iteration. After the auto-bidders' response, the revenue starts to recover and ultimately, the revenue impacts of all boosts become positive after convergence. Again, benchmark boosts outperform uniform boosts in terms of revenue.

\begin{table}[h]
    \centering
    \begin{tabular}{@{}ccccccc@{}}
        \toprule
        \multirow{2}{*}{Boosts} & & \multicolumn{2}{c}{VCG} & & \multicolumn{2}{c}{GSP} \\
        \cline{3-4} \cline{6-7}
         & & Welfare & Revenue & & Welfare & Revenue \\
        \midrule
        uboost-$0.3$ & & $+1.83\%$ & $+1.44\%$ & & $+2.04\%$ & $+1.66\%$ \\
        uboost-$0.6$ & & $+2.25\%$ & $+1.67\%$ & & $+2.38\%$ & $+1.83\%$ \\
        uboost-$0.9$ & & $+2.22\%$ & $+1.50\%$ & & $+2.23\%$ & $+1.58\%$ \\
        uboost-$1.2$ & & $+2.03\%$ & $+1.23\%$ & & $+1.95\%$ & $+1.19\%$ \\
        uboost-$1.5$ & & $+1.78\%$ & $+0.91\%$ & & $+1.65\%$ & $+0.79\%$ \\
        \midrule
        benchmark-$0.3$ & & $+5.99\%$ & $+5.58\%$ & & $+6.24\%$ & $+5.83\%$ \\
        benchmark-$0.6$ & & $+7.28\%$ & $+6.67\%$ & & $+7.42\%$ & $+6.79\%$ \\
        benchmark-$0.9$ & & $+7.81\%$ & $+7.00\%$ & & $+7.86\%$ & $+7.00\%$ \\
        benchmark-$1.2$ & & $+8.06\%$ & $+7.09\%$ & & $+8.06\%$ & $+7.03\%$ \\
        benchmark-$1.5$ & & $+8.20\%$ & $+7.08\%$ & & $+8.16\%$ & $+6.99\%$ \\
        \bottomrule
    \end{tabular}
    \caption{Welfare and revenue lifts of uniform boosts and benchmark boosts after convergence.}
    \label{tab:final}
\end{table}

Table~\ref{tab:final} shows the welfare and revenue impact of uniform boosts and benchmark boosts (with more boost weights) after auto-bidders converge to an equilibrium, and as expected, benchmark boosts outperform uniform boosts significantly. Notice that the welfare performance of uniform boosts starts to drop when the boost weight is large. The reason is that uniform boosts generally do not align with the the optimal allocation maximizing the liquid welfare so that it may have negative impact on welfare. In contrast, the welfare performance of benchmark boosts continues to increase even when the boost weight is large. 

However, the revenue performance does not always increase: it drops when the boost weights increase from $1.2$ to $1.5$. Moreover, observe that there is a gap between the welfare improvement and revenue improvement as shown in Table~\ref{tab:final}, and as the boost weight $c$ increases, the gap becomes larger. This gap mainly comes from the auto-bidders who cannot hit their target spends. Here, we say an auto-bidder hits her target spend if she exhausts her budget or her total payment is equal to her total received values. Note that if all bidders could hit target spends, the welfare and the revenue would be the same. However, in practice, there are auto-bidders who cannot hit their target spends due to the discontinuity of the bidding landscape. In other words, a bidder may need to increase her bid a lot to win one more slot; however, doing so 
might result in violations of either the budget constraint or the return on ad spend constraint. As the boost weight increases, the bidding landscape becomes more and more discontinuous, making it harder and harder for auto-bidders to hit their target spends. Particularly, as the boost weight $c$ approaches $\infty$, the revenue would approach $0$. Therefore, one must be cautious about choosing the boost weight $c$ to balance its impact on welfare and revenue performance.

\section{Conclusions}

In this paper, we propose using auctions with boosts to improve the welfare and revenue for environments with target CPA and target ROAS auto-bidders in VCG auctions. Under the assumption that the auto-bidders adopt uniform bidding strategies, our results can be further extended to GSP auctions and first-price auctions. Our empirical findings demonstrate the effectiveness of our mechanisms, which support the theoretic results. The findings also emphasize that the practitioner must be cautious on the choice of the boost weight $c$ to balance the impact of welfare and revenue.

\balance

\section{Acknowledgement}

Thanks are due to Roozbeh Ebrahimi, S\'{e}bastien Lahaie, Renato Paes Leme, Shaohua Sun, and Ying Wang for helpful discussions, and for their comments and suggestions.

\bibliographystyle{ACM-Reference-Format}
\bibliography{bib}


\end{document}